\documentclass{amsart}
\usepackage{amssymb, stmaryrd, bm}
\usepackage{enumerate}
\usepackage{hyperref}

\newcommand{\A}{\mathbb{A}}
\newcommand{\R}{\mathbb{R}}

\newcommand{\isoto}{\widetilde{\to}}
\newcommand{\vect}{\mathfrak{X}}
\newcommand{\id}{\mathrm{id}}

\newcommand{\pair}[2]{\langle #1, #2 \rangle}
\newcommand{\cbrack}[2]{\llbracket #1, #2 \rrbracket}
\newcommand{\cD}{\mathcal{D}}

\newcommand{\cM}{\mathcal{M}}
\newcommand{\cO}{\mathcal{O}_\cM}

\DeclareMathOperator{\Sh}{Sh}
\DeclareMathOperator{\sgn}{sgn}
\newcommand{\cupprod}{\!\smallsmile\!}

\newcommand{\lie}{\mathcal{L}}
\newcommand{\bnabla}{\bm{\nabla}}
\DeclareMathOperator{\Der}{Der}

\DeclareMathOperator{\End}{End}
\DeclareMathOperator{\im}{im}
\DeclareMathOperator{\tr}{tr}
\newcommand{\cs}{\mathrm{cs}}
\newcommand{\ch}{\mathrm{ch}}
\DeclareMathOperator{\coker}{coker}
\newcommand{\topp}{\mathrm{top}}
\DeclareMathOperator{\sym}{Sym}
\DeclareMathOperator{\rank}{rk}
\newcommand{\ad}{\mathrm{ad}}

\newtheorem{thm}{Theorem}[section]
\newtheorem{prop}[thm]{Proposition}
\newtheorem{lemma}[thm]{Lemma}
\newtheorem{cor}[thm]{Corollary}

\theoremstyle{definition}
\newtheorem{definition}[thm]{Definition}

\theoremstyle{remark}
\newtheorem{remark}[thm]{Remark}

\newtheorem{example}[thm]{Example}

\numberwithin{equation}{section}

\begin{document}

\title[Courant cohomology, Cartan calculus]{Courant cohomology, Cartan calculus, connections, curvature, characteristic classes}
\author{Miquel Cueca}
\address{Mathematics Institute\\Georg-August-University of G\"ottingen\\Bunsenstra{\ss}e 3-5\\G\"ottingen 37073\\Germany}
\email{micueten@impa.br}
\author{Rajan Amit Mehta}
\address{Department of Mathematics and Statistics\\  Smith College\\ 44 College Lane, Northampton, MA 01063\\ USA}
\email{rmehta@smith.edu}

\begin{abstract}
We give an explicit description, in terms of bracket, anchor, and pairing, of the standard cochain complex associated to a Courant algebroid. In this formulation, the differential satisfies a formula that is formally identical to the Cartan formula for the de Rham differential. This perspective allows us to develop the theory of Courant algebroid connections in a way that mirrors the classical theory of connections. Using a special class of connections, we construct secondary characteristic classes associated to any Courant algebroid.
\end{abstract}



\maketitle

\section{Introduction}

Courant algebroids were introduced by Liu, Weinstein, and Xu \cite{lwx:courant}, axiomatizing the properties of brackets studied by Courant and Weinstein \cite{cw:courant, courant:dirac} and Dorfman \cite{dorfman} in the context of Dirac constraints. More recently, Courant algebroids have also appeared in the context of generalized geometry \cite{gualtieri:gcg}, double field theory \cite{deser-stasheff,hull-zwiebach}, and AKSZ sigma models \cite{aksz, roytenberg:aksz}. 

Associated to any Courant algebroid is a cochain complex, known as the \emph{standard complex}. The existence of the standard complex arises immediately from the correspondence, due to \v{S}evera \cite{severa:some} and Roytenberg \cite{dmitry:graded}, between Courant algebroids and degree $2$ symplectic dg-manifolds. In some special cases, such as exact Courant algebroids, the corresponding symplectic dg-manifold can be described explicitly. However, in general, the symplectic dg-manifold associated to a Courant algebroid $E \to M$ is defined implicitly as the minimal symplectic realization of $E[1]$, and explicit formulas for the standard complex and its differential are only available in local coordinates. This difficulty was nicely described by Ginot and Grutzmann \cite{ginot-grutzmann}, who wrote that the standard cohomology of a Courant algebroid ``is quite different from the usual cohomology theories \textellipsis where the cohomology is defined using a differential given by a Cartan-type formula.''

As a way of circumventing the above difficulties, Sti\'enon and Xu \cite{stienon-xu:modular} defined the \emph{na\"ive complex} of a Courant algebroid. They proved that, in degree $1$, the na\"ive cohomology is isomorphic to the standard cohomology; this result was sufficient for their construction of the modular class. In the case of a transitive Courant algebroid, Ginot and Grutzmann \cite{ginot-grutzmann} proved that the na\"ive cohomology is isomorphic to the standard cohomology.  However, for general Courant algebroids, the two cohomologies are different.

The first main result of this paper is to show that there is indeed a description of the standard complex for which the differential has a Cartan formula. The tradeoff is that we need to allow for cochains that only satisfy $C^\infty(M)$-multilinearity and skew-symmetry up to terms involving the bilinear form on the Courant algebroid. Such cochains appeared in the work of Keller and Waldmann \cite{keller-waldmann}, in an algebraic setting where the \v{S}evera-Roytenberg correspondence does not apply. Our contribution is to show that, in the smooth setting, the Keller-Waldmann complex is isomorphic to the standard complex (Theorem \ref{thm:isomorphism}) and that the differential satisfies a Cartan formula (Theorem \ref{thm:cartan}). 

This result is also closely related to work of Roytenberg \cite{roytenberg:dorfman}, where another cochain complex is constructed in an algebraic setting. In this complex, a cochain consists of a sequence of maps, satisfying certain compatibility conditions. Roytenberg showed that, in the smooth setting, this complex is isomorphic to the standard complex. In light of this, our result could be viewed as the further observation that, in the smooth setting, Roytenberg's complex is isomorphic to the Keller-Waldmann complex.

The validity of a Cartan formula provides a way to relate Courant algebroid cohomology to classical (e.g.\ de Rham or Chevalley-Eilenberg) cohomology theories, so that many proofs and calculations can carry over verbatim. As an application of this idea, we consider $E$-connections on $B$, where $E \to M$ is a Courant algebroid and $B \to M$ is a vector bundle. The notion of $E$-connection first appeared in an unpublished manuscript of Alekseev and Xu \cite{alekseev-xu}, and has since been used in numerous contexts, e.g.\ \cite{garcia-fernandez, gualtieri:branes, jurco}. In some of these papers, curvature is introduced, but we could find nowhere in the literature where curvature is interpreted as an $\End(B)$-valued $2$-cochain.

From our new vantage point, we see, in Section \ref{sec:connections}, that an $E$-connection $\nabla$ on $B$ corresponds to a covariant derivative operator $D_\nabla$ on the space of $B$-valued cochains, the curvature $F_\nabla$ is an $\End(B)$-valued $2$-cochain (Proposition \ref{prop:curv}), and that the Bianchi identity holds (Proposition \ref{prop:bianchi}). The Cartan formula allows for proofs that are formally identical to those in the classical theory of connections. 

In Section \ref{sec:modular}, we describe a construction of the modular class, essentially following Sti\'enon and Xu \cite{stienon-xu:modular}. But now, in light of our earlier results, we can interpret the calculations as taking place in the standard complex, rather than the na\"ive complex. We then prove that Courant algebroids are always unimodular (Proposition \ref{prop:modtensor}), simultaneously generalizing the classical result that quadratic Lie algebras are unimodular (e.g.\ \cite{meinrenken:book}) and the result of Sti\'enon and Xu \cite{stienon-xu:modular} that the double of a Lie bialgebroid is unimodular.

As a further application, we construct in Section \ref{sec:characteristic} higher secondary characteristic classes of Courant algebroids. This construction relies on an $E$-connection $\nabla^E$ on $E$ that can be naturally defined using a linear connection on $E$. The formula for $\nabla^E$ closely resembles the ``basic Dorfman connection'' in \cite{madeleine}, but with an additional term that corrects the failure to be $C^\infty(M)$-linear in the first entry. In general, $\nabla^E$ is not flat, but its primary cocycles $\tr(F_\nabla^k)$ vanish when $k$ is odd, and as a result secondary classes can be obtained. In the case of exact Courant algebroids, these secondary classes vanish.

Although characteristic classes can be constructed by other means, e.g.\ \cite{crainic-fernandes:secondary, m-gs:vbalgbds, kotov-strobl, lms:characteristic}, this approach is remarkable because it doesn't require sophisticated machinery such as representations up to homotopy or graded geometry, but instead follows a path that is similar to the classical construction of Chern and Simons \cite{chern-simons}.

In Section \ref{sec:dirac}, we conclude the paper with a few brief remarks about Dirac structures in relation to cohomology and characteristic classes.

\subsection*{Acknowledgements}

We would like to thank Henrique Bursztyn, Rui Loja Fernandes, Madeleine Jotz, Mathieu Sti\'enon, and Ping Xu for helpful conversations related to the paper. M.C. would also like to thank Smith College for hospitality during two visits. M.C. was partially supported by a Ph.D. grant given by CNPq. We would also like to thank the anonymous referee for suggestions that significantly improved the quality of the paper.

\section{Courant cohomology and the Cartan formula}

\subsection{The Keller-Waldmann algebra}
Let $E \to M$ be a vector bundle equipped with a nondegenerate symmetric bilinear form $\pair{\cdot}{\cdot}$. To such a structure we can associate a graded algebra, which we call the \emph{Keller-Waldmann algebra}. This algebra was defined by Keller and Waldmann \cite{keller-waldmann} in an algebraic setting.

\begin{definition}
\label{cochain}
A \emph{$k$-cochain} on $E$, $k \geq 1$, is a map 
\[
\omega: \underbrace{\Gamma(E) \times \cdots \times \Gamma(E)}_{k} \to C^\infty(M)
\]
that is $C^\infty(M)$-linear in the last entry and, for $k \geq 2$, such that there exists a map
\[
\sigma_\omega: \underbrace{\Gamma(E) \times \cdots \times \Gamma(E)}_{k-2} \to \vect(M),
\]
called the \emph{symbol}, such that
\begin{align*}
&\omega(e_1, \dots, e_i, e_{i+1}, \dots, e_k) + \omega(e_1, \dots, e_{i+1}, e_i, \dots, e_k) \\
&= \sigma_\omega (e_1, \dots, \widehat{e}_i, \widehat{e}_{i+1}, \dots, e_k)(\pair{e_i}{e_{i+1}})
\end{align*}
for all $e_j \in \Gamma(E)$ and $1 \leq i \leq k-1$.
\end{definition}
We write $C^k(E)$ to denote the space of $k$-cochains. By definition, we set $C^0(E) = C^\infty(M)$.

\begin{remark}
If $E$ has positive rank, then the map $\sigma_\omega$, which controls the failure of $\omega$ to be skew-symmetric, is uniquely determined by $\omega$. Observe that $\omega\in C^k(E)$ is such that $\sigma_\omega=0$ if and only if $\omega$ is $C^\infty(M)$-linear in each entry and totally skew-symmetric. Thus there is a natural inclusion map $i: \bigwedge^k \Gamma(E^*) \hookrightarrow C^k(E)$.
\end{remark}

In low degrees, we can give simple descriptions of $C^k(E)$:
\begin{itemize}
	\item $C^0(E)=C^\infty(M)$.
	\item $C^1(E)=\Gamma(E^*)$, which can be naturally identified with $\Gamma(E)$ via the bilinear form.
	\item Given $\omega \in C^2(E)$, let $\widehat{\omega}: \Gamma(E) \to \Gamma(E)$ be given by 
	\[\pair{\widehat{\omega}(e_1)}{e_2}=\omega(e_1,e_2).\]
	Then $\widehat{\omega}$ is a covariant differential operator (CDO) with symbol $\sigma_\omega$, i.e.\ it satisfies
	\[ \widehat{\omega}(fe) = \sigma_\omega(f) e + f\widehat{\omega}(e)\]
	for all $f \in C^\infty(M)$ and $e \in \Gamma(E)$, and it is skew-symmetric, i.e.\ it satisfies
	\[ \sigma_\omega \pair{e_1}{e_2} = \pair{\widehat{\omega}(e_1)}{e_2} + \pair{e_1}{\widehat{\omega}(e_2)}\]
	for all $e_1, e_2 \in \Gamma(E)$. The CDOs are the sections of a Lie algebroid $\A_E$, known as the \emph{Atiyah algebroid}, and the skew-symmetric CDOs are the sections of a Lie subalgebroid, which we denote $\A_E^{\pair{}{}}$. It is a simple exercise to show that the map $\omega \mapsto \hat{\omega}$ gives an isomorphism $C^2(E) \cong \Gamma(\A_E^{\pair{}{}})$.
\end{itemize}

The space of cochains $C^\bullet(E)$ is a graded-commutative algebra, where the product is given by
	\[
		(\omega\cupprod\tau)(e_1,\cdots, e_{k+m})=\sum_{\pi\in \Sh(k,m)}\sgn(\pi)\ \omega(e_{\pi(1)},\cdots,e_{\pi(k)})\tau(e_{\pi(k+1)},\cdots, e_{\pi(k+m)})
	\]
	for $\omega\in C^k(E)$ and $\tau\in C^m(E)$; see \cite[Corollary 3.17]{keller-waldmann}\footnote{The formula in \cite[Corollary 3.17]{keller-waldmann} has an extra factor of $(-1)^{km}$, due to a slightly different choice of convention.}.

The following was proven in \cite[Corollary 5.11]{keller-waldmann} under a hypothesis that always holds in the smooth setting.
\begin{prop}\label{KWcomplex}
$C^\bullet(E)$ is generated by $C^0(E), C^1(E)$ and $C^2(E)$ as an algebra.
\end{prop}

Keller and Waldmann also define a degree $-2$ Poisson bracket on $C^\bullet(E)$, making it into a graded Poisson algebra \cite[Theorem 3.18]{keller-waldmann}. We will discuss this bracket further in Remark \ref{rmk:poisson}.

\subsection{Graded symplectic manifolds of degree 2}

We recall the definition of degree $2$ symplectic manifolds. See \cite{cattaneo-schatz, dmitry:graded} for more details.

A \emph{degree $2$ manifold} $\cM$ is a pair $(M, \cO)$, where $M$ is a smooth manifold and $\cO$ is a sheaf of graded commutative algebras such that, for all $p\in M$, there exists a neighborhood $U$ such that $\cO(U)$ is generated by homogeneous coordinates $\{x^i, e^i, p^i\}$, where $x^i$ are coordinates on $M$, and where $e^i$ and $p^i$ are coordinates of degree $1$ and $2$, respectively.

A degree $2$ manifold $\cM$ is \emph{symplectic} if $\cO$ has a degree $-2$ Poisson bracket such for which there exist local homogeneous coordinates $\{x^i, e^i, p^i\}$ such that $\{x^i,p^j\}$ and $\{e^i,e^j\}$ are invertible matrices; for a more detailed definition, see \cite{cattaneo-schatz}.

\begin{prop}[\cite{dmitry:graded}]\label{prop:deg2symplectic}
There is a one-to-one correspondence between symplectic degree $2$ manifolds and vector bundles equipped with a nondegenerate symmetric bilinear form.
\end{prop}
Briefly, the correspondence is as follows. Given a vector bundle $E \to M$ with a nondegenerate symmetric bilinear form $\pair{\cdot}{\cdot}$, let $\cO^0 = C^\infty(M)$, $\cO^1 = \Gamma(E)$, and $\cO^2 = \Gamma(\A_E^{\pair{}{}})$. In low degrees, the nonvanishing Poisson brackets are
\begin{align*}
 \{e_1, e_2\} &= \pair{e_1}{e_2}, & \{\phi, f\} &= \sigma_\phi(f),\\
 \{\phi, e\} &= \phi(e), & \{\phi_1, \phi_2\} &= \phi_1 \phi_2 - \phi_2 \phi_1
 \end{align*}
 for $f \in \cO^0$, $e, e_1, e_2 \in \cO^1$, and $\phi, \phi_1, \phi_2 \in \cO^2$. Since $\cO$ is generated in degrees $0$, $1$, and $2$, this information is sufficient to determine $\cO$, albeit in a not very explicit way.

Comparing the descriptions of $C^k(E)$ and $\cO^k$ in low degrees, we immediately see that the two are naturally isomorphic for $k \leq 2$. The following theorem shows that this isomorphism extends to all degrees.

\begin{thm}\label{thm:isomorphism}
	Let $E\to M$ be a vector bundle of positive rank, equipped with a nondegenerate symmetric bilinear form. Let $\cM$ be the corresponding symplectic degree $2$ manifold. Then the map $\Upsilon: \cO^k \to C^k(E)$, given by
\[ \Upsilon(\psi)(e_1,\dots,e_k)=\{e_k,\{e_{k-1},\dots,\{e_1,\psi\}\cdots\}, \]
is an isomorphism of graded commutative algebras.
\end{thm}

\begin{proof}
First we show that $\Upsilon$ is well-defined. Since
\begin{align*}
\Upsilon(\psi)(e_1,\dots, fe_k)& = \{fe_k,\{e_{k-1},\dots,\{e_1,\psi\}\cdots \} \\
&=f\{e_k,\{e_{k-1},\dots,\{e_1,\psi\}\cdots \} \\
&=f\Upsilon(\psi)(e_1,\dots, e_k)
\end{align*}
for $\psi \in \cO^k$, $f \in C^\infty(M)$, $e_i \in \Gamma(E)$, we see that $\Upsilon$ is $C^\infty(M)$-linear in the last entry.

Let $\sigma_{\Upsilon(\psi)}$ be given by
\[ \sigma_{\Upsilon(\psi)}(e_1, \dots, e_{k-2})(f) = \{f,\{e_{k-2}, \dots, \{e_1, \psi\} \cdots \} . \]
Then, using the Jacobi identity and the fact that certain brackets vanish by degree considerations, we have
\begin{align*}
&\Upsilon(\psi)(e_1,\dots, e_i, e_{i+1},\dots, e_k)+\Upsilon(\psi)(e_1,\dots, e_{i+1}, e_{i},\dots, e_k)\\
&=	\{e_k,\dots,\{e_{i+1},\{e_{i},\dots,\{e_1,\psi\}\cdots\}+\{e_k,\dots,\{e_{i},\{e_{i+1},\dots,\{e_1,\psi\}\cdots\}\\
&=	\{e_k,\dots,\{\{e_{i+1},e_i\},\{e_{i-1},\dots,\{e_1,\psi\}\cdots\}\\
&= \{\{e_{i+1},e_i\},\{e_k,\dots,\{e_1,\psi\}\cdots\}\\
&=	\sigma_{\Upsilon(\psi)}(e_1,\dots,\widehat{e}_i,\widehat{e}_{i+1},\dots, e_k)(\pair{e_i}{e_{i+1}}).
\end{align*}
This shows that $\Upsilon(\psi)$ is indeed an element of $C^k(E)$, so $\Upsilon$ is well-defined.

For $\psi \in \cO^k$, $\eta \in \cO^m$, repeated application of the Leibniz rule gives
\begin{align*}
	&\Upsilon(\psi \eta)(e_1,\dots,e_{k+m})=\{e_{k+m},\dots,\{e_1,\psi \eta\}\} \\
	& =\{e_{k+m},\dots,\{e_1, \psi\}\eta+(-1)^{|\psi|}\psi\{e_1,\eta\}\} \\
		&= \sum_{\pi\in \Sh(k,m)}\sgn(\pi)\{e_{\pi(k)},\dots,\{e_{\pi(1)},\psi\}\}\{e_{\pi(k+m)},\dots,\{e_{\pi(k+1)},\eta\}\} \\
		&= \sum_{\pi\in \Sh(k,m)}\sgn(\pi)\Upsilon(\psi)(e_{\pi(1)},\dots,e_{\pi(k)})\Upsilon(\eta)(e_{\pi(k+1)},\dots,e_{\pi(k+m)}) \\
		&= (\Upsilon(\psi)\cupprod\Upsilon(\eta))(e_1,\dots, e_{k+m}).
\end{align*}
This shows that $\Upsilon$ is a morphism of algebras. By Proposition \ref{KWcomplex} and the fact that $\Upsilon$ is an isomorphism in degrees $k \leq 2$, we deduce that $\Upsilon$ is onto.

It can be shown that $\Upsilon$ is one-to-one by induction on $k$. We have already established the base cases $k \leq 2$. Suppose that $\psi \in \cO^k$ is such that $\Upsilon(\psi) = 0$. Then, for any $e, e_1, \dots, e_{k-1} \in \Gamma(E)$, we see that
\[ \Upsilon(\{e,\psi\})(e_1, \dots, e_{k-1}) = \Upsilon(\psi)(e, e_1, \dots, e_{k-1}) = 0,\]
so by the inductive hypothesis we have $\{e, \psi\} = 0$ for all $e \in \Gamma(E)$. By a local coordinate argument, it follows that $\psi = 0$. 
\end{proof}

\begin{remark}\label{rmk:poisson}
The isomorphism $\Upsilon$ allows us to transfer the degree $-2$ Poisson bracket from $\cO$ to $C^\bullet(E)$. By repeated application of the Jacobi identity, one could derive an explicit formula for $[\omega, \tau](e_1, \dots e_{k+m-2})$, where $\omega \in C^k(E)$, $\tau \in C^m(E)$, as sums over (un)shuffles of terms involving $\omega$ and $\tau$. For the purposes of this paper, such a formula is not needed.
\end{remark}

\begin{remark}\label{split}
Keller and Waldmann \cite{keller-waldmann} also consider the \emph{Rothstein algebra} 
\[ \mathcal{R}(E) = \sym(\vect(M)) \otimes \bigwedge \Gamma(E).\] 
They show that, given a choice of metric connection on $E$, one can put a degree $-2$ Poisson bracket on $\mathcal{R}(E)$, and that, under hypotheses that always hold in the smooth case, there is an isomorphism $\mathcal{R}(E) \isoto C(E)$. This isomorphism depends on the choice of connection, and it is constructed using iterated brackets similar to those in the proof of Theorem \ref{thm:isomorphism}.

It is also well-known from the perspective of graded geometry \cite{dmitry:graded} that $\mathcal{R}(E)$ is noncanonically isomorphic to $\cO$, with the isomorphism depending on a choice of connection. The result of Theorem \ref{thm:isomorphism} thus closes the diagram with a \emph{canonical} isomorphism $\cO \isoto C(E)$.
\end{remark}

\begin{example}\label{ex-VE}
Consider the case where $M$ is a point, so that $E=V$ is a vector space equipped with a nondenerate symmetric bilinear form. In this case, the Keller-Waldmann cochains are necessarily skew-symmetric, so
\[ C(V) = \bigwedge V^*.\]
Although the bilinear form does not affect the algebra structure in this case, it is used to define the degree $-2$ Poisson bracket.
\end{example}

\begin{example}\label{ex-aa}
Let $A \to M$ be a vector bundle, and let $E = A \oplus A^*$ with the bilinear form
\[ \pair{X_1 + \xi_1}{X_2 + \xi_2} = \xi_1(X_2) + \xi_2(X_1).\]
In this case, the Keller-Waldmann algebra can be identified with the algebra of polyderivations of $\mathcal{A} = \bigwedge \Gamma(A^*)$, with an appropriate degree shift:
\[ C(E) \cong \sym\left( \Der(\mathcal{A})[-2]\right).\]
Here, the symmetric product is taken as a module over $\bigwedge \Gamma(A^*)$. We give here a brief sketch of this correspondence, leaving some details to the reader.

For $\beta \in \mathcal{A}_k = \bigwedge^k \Gamma(A^*)$, the corresponding $k$-cochain $c(\beta) \in C^k(E)$ is given by
\[ c(\beta)(X_1 + \xi_1, \dots, X_k + \xi_k) = \beta(X_1, \dots , X_k).\]

The degree $-1$ derivations of $\mathcal{A}$ are the contraction operators $\iota_Y$, where $Y \in \Gamma(A)$. The corresponding $1$-cochain $I_Y \in C^1(E)$ is given by 
\[I_Y(X+\xi) = \xi(Y).\]

The degree $0$ derivations of $\mathcal{A}$ correspond to covariant differential operators on $A$. Given a CDO $\phi$ with symbol $\sigma_\phi$, the corresponding $2$-cochain $\widetilde{\phi} \in C^2(E)$ is given by 
\[ \widetilde{\phi}(X_1 + \xi_1, X_2 + \xi_2) = \xi_2(\phi(X_1)) - \xi_1(\phi(X_2)) + \sigma_\phi(\xi_1(X_2)).\]
Since $\Der(\mathcal{A})$ is generated as an $\mathcal{A}$ module in degrees $-1$ and $0$, the above correspondences are sufficient to determine the isomorphism $C(E) \cong \sym\left( \Der(\mathcal{A})[-2]\right)$. 

In light of Theorem \ref{thm:isomorphism}, the above correspondence agrees with the graded geometry perspective, where the symplectic degree $2$ manifold is $T^*[2]A[1]$, whose functions are polyvector fields on $A[1]$.
\end{example}

\begin{example}\label{ex-ST}
An important special case of Example \ref{ex-aa} is when $E = TM \oplus T^*M$. Then
\[ C(E) \cong \sym\left(\Der(\Omega(M))[-2]\right).\]
Furthermore, using the results of \cite{fro:nij} on the algebra of derivations of $\Omega(M)$ we can obtain an isomorphism
\[  C^{k}(E)\cong \bigoplus_{i+j+2l=k}\Omega^i\left(M;\bigwedge\nolimits^ jTM\otimes \sym^l TM\right).\]
We note that the iterated bracket construction combined with the nonlinear isomorphism between derivations and vector valued forms makes this isomorphism highly non-trivial, albeit canonical.
\end{example}

\subsection{Courant algebroids}

\begin{definition}\label{courant}
A \emph{Courant algebroid} is a vector bundle $E \to M$ equipped with a nondegenerate symmetric bilinear form $\pair{\cdot}{\cdot}$, a bundle map $\rho: E \to TM$ (called the \emph{anchor}), and a bracket $\cbrack{\cdot}{\cdot}$
(called the \emph{Courant bracket}) such that
\begin{enumerate}[(C1)]
\item $\cbrack{e_1}{fe_2} = \rho(e_1)(f) e_2 + f\cbrack{e_1}{e_2}$,
\item $\rho(e_1)(\pair{e_2}{e_3}) = \pair{\cbrack{e_1}{e_2}}{e_3} + \pair{e_2}{\cbrack{e_1}{e_3}}$,
\item $\cbrack{\cbrack{e_1}{e_2}}{e_3} = \cbrack{e_1}{\cbrack{e_2}{e_3}} - \cbrack{e_2}{\cbrack{e_1}{e_3}}$,
\item $\cbrack{e_1}{e_2} + \cbrack{e_2}{e_1} = \cD \pair{e_1}{e_2}$,
\end{enumerate}
for all $f \in C^\infty(M)$ and $e_i \in \Gamma(E)$, where $\cD : C^\infty(M) \to \Gamma(E)$ is defined by
\begin{equation*}
\pair{\cD f}{e} = \rho(e)(f).
\end{equation*}
\end{definition}
Note that we are using the ``Dorfman convention'', where the bracket is not skew-symmetric (at least if $\rho$ is nonzero), but where a Jacobi identity (axiom (C3)) holds.

\begin{remark}\label{rmk:anchor}
The following identities, which were axioms in the original definition of Courant algebroid \cite{lwx:courant}, are consequences of the axioms in Definition \ref{courant}:
\begin{enumerate}[({A}1)]
    \item $\rho(\cbrack{e_1}{e_2}) = [\rho(e_1),\rho(e_2)]$ for all $e_1, e_2 \in \Gamma(E)$,
    \item $\rho \circ \cD = 0$.
\end{enumerate}
See \cite{uchino} for a discussion of the dependencies among the axioms.
\end{remark}	

Let $E \to M$ be a Courant algebroid. Since $E$ is, in particular, a vector bundle equipped with a nondegenerate symmetric bilinear form, there is by Proposition \ref{prop:deg2symplectic} an associated symplectic degree $2$ manifold $\cM$. The following theorem, due to \v{S}evera \cite{severa:some} and Roytenberg \cite{dmitry:graded}, goes further to state that the anchor and bracket can be completely encoded in a certain degree $3$ function on $\cM$.
\begin{thm}[\cite{dmitry:graded, severa:some}]\label{thm:deg3}
Let $E \to M$ be a vector bundle equipped with a nondegenerate symmetric bilinear form, with corresponding symplectic degree $2$ manifold $\cM$. Then there is a one-to-one correspondence between anchor-and-bracket data $(\rho, \cbrack{\cdot}{\cdot})$ satisfying the axioms of Definition \ref{courant} and degree $3$ functions $\theta \in \cO^3$ such that $\{\theta,\theta\} = 0$.
\end{thm}	

\begin{remark}\label{couformulas}
The correspondence in Theorem \ref{thm:deg3} is given by the following derived bracket formulas (see \cite{dmitry:graded}):
\begin{align}
\rho(e)(f) &= \{\{e, \theta\}, f\} = -\{f, \{e, \theta\}\}, \label{eqn:derivedanchor} \\
\cbrack{e_1}{e_2} &= \{\{e_1,\theta\},e_2\} = \{e_2,\{e_1,\theta\}\}, \label{eqn:derivedbracket}
\end{align}
for $e, e_1, e_2 \in \cO^1 = \Gamma(E)$ and $f \in \cO^0 = C^\infty(M)$. As a consequence, we have the following corollary.
\end{remark}
\begin{cor}\label{cor:T}
Let $E \to M$ be a Courant algebroid with associated symplectic degree $2$ manifold $\cM$ and degree $3$ function $\theta$. Let $T = \Upsilon(\theta) \in C^3(E)$. Then $T$ is given by
		\begin{equation*}
			T(e_1,e_2,e_3)=\pair{\cbrack{e_1}{e_2}}{e_3}.
		\end{equation*}
\end{cor}

\subsection{Cartan calculus}

Let $E \to M$ be a Courant algebroid with associated symplectic degree $2$ manifold $\cM$ and degree $3$ function $\theta$. Since $\theta$ satisfies the classical master equation $\{\theta,\theta\}=0$, the operator $d_E=\{\theta,\cdot\}$ defines a differential on $\cO$. The \emph{standard complex} of a Courant algebroid is defined to be the cochain complex $(\cO, d_E)$.

For any section $e \in \Gamma(E)$, we can introduce the ``contraction'' and ``Lie derivative'' operators 
\begin{align*}
\iota_e &= \{e,\cdot\}, & \lie_e &=  \{\{e,\theta\},\cdot\},
\end{align*}
on $\cO$. The contractions, Lie derivatives, and differential satisfy many relations that are formally identical to the Cartan relations in de Rham theory.
\begin{prop}\label{prop:cartan}
The following graded commutation relations hold for all $e, e' \in \Gamma(E)$:
\begin{align*}
d_E^2 &= 0, & [\lie_e, d_E] & = 0,\\
[\iota_e, d_E] &= \lie_e, & [\lie_{e}, \lie_{e'}] &= \lie_{\cbrack{e}{e'}},\\
[\lie_{e}, \iota_{e'}] &= \iota_{\cbrack{e}{e'}}
\end{align*}
\end{prop}
\begin{proof}
The proofs involve using the Jacobi identity and known identities for Poisson brackets involving $\theta$ and sections of $E$. For example, using the Jacobi identity and \eqref{eqn:derivedbracket}, we have
\begin{align*}
[\lie_{e}, \iota_{e'}] &= \lie_{e} \iota_{e'} - \iota_{e'} \lie_{e} \\
&= \{\{e,\theta\}, \{e', \cdot\}\} - \{e', \{\{e,\theta\}, \cdot \}\} \\
&= \{\{\{e, \theta\}, e'\}, \cdot \} \\
&= \iota_{\cbrack{e}{e'}}.
\end{align*}
We leave the other relations as exercises for the reader.
\end{proof}
\begin{remark}
It should be emphasized that contraction operators do not anticommute with each other, in contrast to the case of de Rham theory. Specifically,
\begin{align*}
[\iota_{e},\iota_{e'}] &= \{e, \{e', \cdot\}\} + \{e', \{e, \cdot\}\}\\
&= \{\{e,e'\}, \cdot\},
\end{align*}
which does not vanish in general. This situation could be handled by introducing additional degree $-2$ contraction operators $\iota_f = \{f, \cdot\}$ for $f \in C^\infty(M)$ and using them to extend the Cartan calculus. This idea is closely connected to the $L_\infty$-algebra structure associated to a Courant algebroid \cite{roytenberg-weinstein}. Rather than considering this, we will take the simpler approach of not trying to exchange contraction operators.
\end{remark}

The result of Theorem \ref{thm:isomorphism} allows us to use the isomorphism $\Upsilon$ to transfer the operators $\iota_e$, $\lie_e$, and $d_E$ to the Keller-Waldmann complex $C^k(E)$. The advantage is that we can now interpret $\iota_e$ as a contraction operator in the usual sense, since
\[ (\iota_e \omega)(e_1, \dots, e_{k-1}) = \omega(e, e_1, \dots, e_{k-1})\]
for $\omega \in C^k(E)$.
\begin{thm}\label{thm:cartan}
For $\omega \in C^k(E)$ and $e \in \Gamma(E)$, $\lie_e \omega$ and $d_E \omega$ are given by the following formulas:
\begin{align*}
(\lie_e \omega)(e_1, \dots, e_k) =& \rho(e)\left(\omega(e_1, \dots, e_k)\right) - \sum_{i=1}^k \omega(e_1, \dots, e_{i-1}, \cbrack{e}{e_i}, e_{i+1}, \dots, e_k),\\
(d_E \omega)(e_0, \dots, e_k) =& \sum_{i=0}^k (-1)^i \rho(e_i)\omega(e_0, \dots, \widehat{e}_i, \dots, e_k)  \\
&- \sum_{i < j} (-1)^i \omega(e_0, \dots, \widehat{e}_i, \dots, e_{j-1}, \cbrack{e_i}{e_j}, e_{j+1}, \dots, e_k).
\end{align*}
\end{thm}
\begin{proof}
For $f \in C^\infty(M)$, the formula \eqref{eqn:derivedanchor} gives us $\lie_e f = \rho(e)(f)$. Then, by repeated use of the relation $[\lie_{e}, \iota_{e'}] = \iota_{\cbrack{e}{e'}}$, we have
\begin{align*}
(\lie_e \omega) (e_1, \dots, e_k) &= \iota_{e_k} \cdots \iota_{e_1} \lie_e \omega \\
&= \lie_e \iota_{e_k} \cdots \iota_{e_1} \omega - \sum_{i=1}^k \iota_{e_k} \cdots \iota_{e_{i-1}} \iota_{\cbrack{e}{e_i}} \iota_{e_{i+1}} \cdots \iota_{e_1} \omega \\
&= \rho(e)\left(\omega(e_1, \dots, e_k)\right) - \sum_{i=1}^k \omega(e_1, \dots, e_{i-1},\cbrack{e}{e_i}, e_{i+1}, \dots, e_k).
\end{align*}
Similarly,
\begin{align*}
&(d_E \omega)(e_0, \dots, e_k) = \iota_{e_k} \dots \iota_{e_0} d_E \omega \\
&= \sum_{i=0}^k (-1)^i \lie_{e_i} \iota_{e_k} \cdots \widehat{\iota}_{e_i} \cdots \iota_{e_0} \omega - \sum_{i < j} (-1)^i \iota_{e_k} \cdots \iota_{e_{j-1}}\iota_{\cbrack{e_i}{e_j}} \iota_{e_{j+1}} \cdots \widehat{\iota}_{e_i} \cdots \iota_{e_0} \omega \\
&= \sum_{i=0}^k (-1)^i \rho(e_i)\omega(e_0, \dots, \widehat{e}_i, \dots, e_k)  - \sum_{i < j} (-1)^i \omega(e_0, \dots, \widehat{e}_i, \dots, \cbrack{e_i}{e_j}, \dots, e_k).
\end{align*}
\end{proof}

\begin{remark}
In $\cO$, it is clear that the degree $3$ function $\theta$ is a cocycle, since $d_E \theta = \{\theta, \theta\} = 0$. It follows that $T = \Upsilon(\theta)$ is a cocycle in $C^3(E)$. An alternative approach is to directly define $T$ by the formula in Corollary \ref{cor:T} and then use the Cartan formula to verify that $d_E T = 0$.
\end{remark}

\begin{example}
In the case where $M$ is a point, a Courant algebroid is a Lie algebra $\mathfrak{g}$ equipped with an invariant nondegenerate symmetric bilinear form. Recall (see Example \ref{ex-VE}) that, in this case, $C(\mathfrak{g}) = \bigwedge \mathfrak{g}^*$. The associated $3$-cocycle (see Corollary \ref{cor:T}) is the Cartan $3$-cocycle 
\[T(v_1,v_2,v_3) = \pair{[v_1,v_2]}{v_3},\]
and from Theorem \ref{thm:cartan} we see that the differential $d_\mathfrak{g}$ is the Chevalley-Eilenberg differential.
\end{example}

\begin{example}\label{standard-CA}
The standard Courant algebroid on a manifold $M$ is defined as $E = TM \oplus T^*M$, with the bilinear form
\begin{equation*}
    \pair{X + \alpha}{Y + \beta} = \alpha(Y) + \beta(X)
\end{equation*}
and bracket 
 \begin{equation*}
     \cbrack{X+\alpha}{Y+\beta}=[X,Y]+\lie_{X}\beta -\iota_{Y} d\alpha.
 \end{equation*}
 In Example \ref{ex-ST}, it was observed that, in this case, $C(E) \cong \sym(\Der(\Omega(M))[-2])$. Under this identification, the $3$-cocycle $T \in C^3(E)$ corresponds to the de Rham operator $d$, which is a degree $1$ derivation of $\Omega(M)$.
  
If $H \in \Omega^3(M)$ is a closed $3$-form, then we can use the above identification to view $d+H$ as an element of $C^3(E)$, thus inducing the $H$-twisted version of the above bracket \cite{severa:letters, se-we:twist}, given by
 \begin{equation*}
      \cbrack{X+\alpha}{Y+\beta}_H=[X,Y]+\lie_{X}\beta -\iota_{Y} d\alpha + \iota_{X} \iota_{Y} H.
 \end{equation*}

We note that this perspective provides a good framework for the interpretation \cite{alekseev-xu, ks:derived} of the standard Courant bracket as a Vinogradov bracket.
\end{example}

\begin{example}
In contrast with many cohomology theories, Courant algebroid cohomology is not necessarily bounded in degree. For example, consider a Courant algebroid $E\to M$ with $\dim(M) > 0$ and $\rank(E) > 0$, where the anchor and bracket are zero. Since the differential is zero, we have
\[H^k(E)=C^k(E).\]
Let $\omega\in C^2(E)$ be such that $\sigma_\omega \neq 0$, and let $e\in\Gamma E$ be such that $\sigma_\omega(\langle e,e\rangle) \neq 0$. Such $\omega$ and $e$ always exist in this situation.

From Definition \ref{cochain}, we have 
\[\omega(e,e)=\frac{1}{2}\sigma_\omega(\langle e,e\rangle) \neq 0,\] 
and thus
\[ (\omega\cupprod\omega)(e,e,e,e)=\omega(e,e)^2\sum_{\pi\in Sh(2,2)}\sgn(\pi) = 2\omega(e,e)^2 \neq 0.\]
Similarly, one can show that $\omega^k$ and $e \cupprod \omega^k$ are nonzero, so in this case $C^k(E)$ is nontrivial for all $k \geq 0$. 
\end{example}

\section{Courant algebroid connections}
\label{sec:connections}

Let $E\to M$ be a Courant algebroid. As mentioned in the introduction, the notion of an $E$-connection arose in unpublished work of Alekseev and Xu \cite{alekseev-xu}. In this section, we develop the theory of $E$-connections in light of Theorems \ref{thm:isomorphism} and \ref{thm:cartan}. The key point is that, because of the existence of Cartan formulas, many results can be proved in a way that is virtually identical to the classical theory. We will sometimes omit such proofs, instead focusing on the aspects that are specific to Courant algebroids.

\subsection{Connections, covariant derivatives, and curvature}

Let $B\to M$ be a vector bundle.

\begin{definition}
An \emph{$E$-connection} on $B$ is a map $\nabla:\Gamma (E)\times \Gamma (B)\to \Gamma (B)$ such that
	\begin{enumerate}
		\item $\nabla_e(fb)=f\nabla_e(b)+\rho(e)(f)b$,
		\item $\nabla_{fe}b=f\nabla_e b$,
	\end{enumerate}
	for all $f \in C^\infty(M)$, $e \in \Gamma(E)$, and $b \in \Gamma(B)$.
\end{definition} 

Let $C^\bullet(E;B)=C^\bullet(E)\otimes \Gamma (B)$ denote the space of $B$-valued cochains. If $\nabla$ is an $E$-connection on $B$, then we can define a covariant derivative operator $D_\nabla$ on $C^\bullet(E;B)$ by the formula
\begin{equation} \label{eqn:dcartan}
\begin{split}
    (D_\nabla \tau)(e_0, \dots, e_k) =& \sum_{i=0}^k (-1)^i \nabla_{e_i}\tau(e_0, \dots, \widehat{e}_i, \dots, e_k)  \\
&- \sum_{i < j} (-1)^i \tau(e_0, \dots, \widehat{e}_i, \dots, \cbrack{e_i}{e_j}, \dots, e_k)
\end{split}
\end{equation}
for $\tau \in C^k(E;B)$. The covariant derivative operator satisfies the property
\begin{equation}\label{eqn:dderiv}
 D_\nabla(\omega \cupprod \tau) = (d_E \omega) \cupprod \tau + (-1)^k \omega \cupprod (D_\nabla \tau)
 \end{equation}
for $\omega \in C^k(E)$ and $\tau \in C^\bullet(E;B)$.

Conversely, given a degree $1$ operator $D$ on $C^\bullet(E;B)$ satisfying \eqref{eqn:dderiv}, there exists a unique $E$-connection $\nabla$ such that $D = D_\nabla$.

\begin{definition}
The \emph{curvature} of an $E$-connection $\nabla$  on $B$ is defined as the map $F_\nabla: \Gamma(E) \times \Gamma(E) \times \Gamma(B) \to \Gamma(B)$ given by
\begin{equation}\label{eqn:curv}
 F_\nabla(e_1, e_2)(b) = \nabla_{e_1}\nabla_{e_2}b - \nabla_{e_2} \nabla_{e_1}b - \nabla_{\cbrack{e_1}{e_2}}b.
 \end{equation}
If $\nabla$ is flat, i.e.\ $F_\nabla = 0$, then we say that $\nabla$ is a \emph{representation} of $E$ on $B$.
\end{definition}

\begin{prop}\label{prop:curv}
	The curvature $F_\nabla$ is an element of $C^2(E; \End(B))$. 
\end{prop}
\begin{proof}
We will sketch two different approaches to the proof. One is to use \eqref{eqn:dcartan} to show that $F_\nabla(e_1,e_2)(b) = (D_\nabla^2 b)(e_1,e_2)$. This proves that $F_\nabla(\cdot, \cdot)(b)$ is a $2$-cochain. Using \eqref{eqn:dderiv}, we can also see that $F_\nabla$ is $C^\infty(M)$-linear in $b$, which proves that it takes values in $\End(B)$.

The other approach is to directly check, using \eqref{eqn:curv}, that $F_\nabla(e_1,e_2)(b)$ is $C^\infty(M)$-linear in $e_2$ and $b$, and to observe that 
\begin{equation*}
	 F_\nabla(e,e')(b)+F_\nabla(e',e)(b)=-\nabla_{\cbrack{e}{e'}+\cbrack{e'}{e}}b=-\nabla_{\cD\pair{e}{e'}}b.
\end{equation*}	
Therefore, $F_\nabla$ satisfies the conditions of Definition \ref{cochain} with $\sigma_{F_\nabla}(f)(b) = -\nabla_{\cD f} b$. Here, it is important that $\sigma_{F_\nabla}$ is $C^\infty(M)$-linear in $b$ (because $\rho \circ \cD = 0$) and a derivation in $f$.
\end{proof}

The $E$-connection $\nabla$ naturally induces an $E$-connection $\widetilde{\nabla}$ on $\End(B)$, given by
\begin{equation}\label{eqn:connend}
\widetilde{\nabla}_e \tau = [\nabla_e, \tau]
\end{equation}
for $\tau \in \Gamma(\End(B))$ and $e \in \Gamma(E)$, where the right side is a commutator bracket of operators on $\Gamma(B)$. The following is a straightforward calculation using \eqref{eqn:dcartan} and \eqref{eqn:curv}.
\begin{prop}\label{prop:bianchi}
The Bianchi identity $D_{\widetilde{\nabla}}F_\nabla = 0$ holds.
\end{prop}

Suppose that $B$ is equipped with a nondegenerate bilinear form $\pair{\cdot}{\cdot}_B$. Then we can define the adjoint $E$-connection $\nabla^\dagger$ on $B$, given by
\begin{equation}\label{eqn:adjoint}
\pair{\nabla^\dagger_e b_1}{b_2}_B = \rho(e)\pair{b_1}{b_2}_B - \pair{b_1}{\nabla_e b}_B
\end{equation}
for all $e \in \Gamma(E)$, $b_1,b_2 \in \Gamma(B)$. The following is a straightforward calculation using \eqref{eqn:curv} and (A1) from Remark \ref{rmk:anchor}.
\begin{prop}\label{prop:adjointcurve}
The curvatures of $\nabla$ and $\nabla^\dagger$ are related by the identity $F_{\nabla^\dagger}=-(F_\nabla)^\dagger$.
\end{prop}

\subsection{Examples}
The following are some important examples of $E$-connections.

\begin{example}\label{ex:wedgetop}
There is a canonical representation $\nabla^{\topp}$ of $E$ on $\bigwedge^{\topp}E$, given by 
\[\nabla^{\topp}_e(e_1 \wedge \cdots \wedge e_m) = \sum_{i=1}^m e_1 \wedge \cdots \wedge \cbrack{e}{e_i} \wedge \cdots \wedge e_m.\]
This is the representation that was used by Sti\'enon and Xu \cite{stienon-xu:modular} in their construction of the modular class. We will discuss this further in Section \ref{sec:modular} (also see Remark \ref{rmk:top}).
\end{example}

\begin{example}\label{ex:regular}
If $E$ is regular, i.e.\ if the anchor map $\rho$ has constant rank, then there are canonical representations on the following bundles:
\begin{itemize}
    \item $\ker\rho/\im\rho^*$. This representation is given by $\nabla_e \bar{e}'=\overline{\lie_e e'}=\overline{\cbrack{e}{e'}}$, where the overline indicates the image under the projection $\ker\rho\to \ker\rho/\im\rho^*$.
    \item $\ker \rho^*$. This representation is given by $\nabla_e \alpha = \lie_{\rho(e)}\alpha$ for $\alpha \in \Gamma(\ker \rho^*) \subseteq \Omega^1(M)$.
    \item $\coker \rho$. This representation is given by $\nabla_e \bar{X}=\overline{\lie_{\rho(e)}X}=\overline{[\rho(e),X]}$, where the overline indicates the image under the projection $TM\to \coker \rho$.
\end{itemize}
\end{example}

\begin{example}\label{repCou}
Given a choice of linear connection $\nabla: \vect(M) \times \Gamma(E) \to \Gamma(E)$, we can define the following $E$-connections. These $E$-connections extend the representations in Example \ref{ex:regular} but are not flat in general.
\begin{itemize}
    \item $\nabla^E$ is an $E$-connection on $E$, given by
    \[ \nabla_e^E e'=\cbrack{e}{e'}+\nabla_{\rho(e')}e-\rho^*\pair{D_\nabla e}{e'}.\]
    \item $\nabla^{TM}$ is an $E$-connection on $TM$, given by
    \[ \nabla_e^{TM} X= [\rho(e),X]+\rho(\nabla_X e). \]
    \item $\nabla^{T^*M}$ is an $E$-connection on $T^*M$, given by
    \[\nabla_e^{T^*M} \alpha=\lie_{\rho(e)}\alpha-\pair{D_\nabla e}{\rho^*\alpha}.\]
\end{itemize}

To verify that $\nabla^E$ is an $E$-connection, we check that
\begin{align*}
     \nabla_e^E fe'= & \cbrack{e}{fe'}+\nabla_{\rho(fe')}e-\rho^*\pair{D_\nabla e}{fe'} \\
    = & f\cbrack{e}{e'}+\rho(e)(f)e'+f\nabla_{\rho(e')}e-f\rho^*\pair{D_\nabla e}{e'} \\
    =& f\nabla^E_ee'+\rho(e)(f)e'
\end{align*}
and
\begin{align*}
     \nabla_{fe}^E e'= & \cbrack{fe}{e'}+\nabla_{\rho(e')}fe-\rho^*\pair{D_\nabla fe}{e'} \\
    = & -\cbrack{e'}{fe}+\cD\pair{fe}{e'}+f\nabla_{\rho(e')}e+\rho(e')(f)e-\rho^*\pair{f D_\nabla e}{e'}-\rho^*\pair{df\otimes e}{e'}\\
    =&-f\cbrack{e'}{e}+f\cD\pair{e}{e'}+f\nabla_{\rho(e')}e-f\rho^*\pair{D_\nabla e}{e'}=f\nabla_e^E e'.
\end{align*}
For $\nabla^{TM}$ and $\nabla^{T^*M}$ the computations are similar.
\end{example}

\begin{example}\label{H-Econ}
Consider the case $E = TM \oplus T^*M$, with the $H$-twisted Courant bracket; see Example \ref{standard-CA}. In this case, it is possible to make a special choice of linear connection so that the $E$-connection $\nabla^E$ takes a simple form.  

First, choose a torsion-free linear connection $\nabla$ on $TM$, and let $\nabla^\dagger$ be the dual connection on $T^*M$.
Define a linear connection $\widehat{\nabla}$ on $TM\oplus T^*M$ by the formula 
 \[\widehat{\nabla}_X (Y+\beta)=\nabla_X Y+\nabla^\dagger_X \beta +\frac{1}{2}i_Xi_YH\]
 for $X,Y\in\mathfrak{X}(M),\ \beta\in\Omega^1(M)$. The associated $E$-connection $\nabla^E$ is then given by
 \begin{align*}
        \nabla^E_{X+\alpha}(Y+\beta)=&\cbrack{X+\alpha}{Y+\beta}_H+\widehat{\nabla}_Y (X+\alpha)-\langle D_{\widehat{\nabla}}(X+\alpha), Y+\beta\rangle\\
        =& [X,Y]+\lie_X\beta-i_Yd\alpha+i_Xi_YH+\nabla_Y X+\nabla_Y^\dagger \alpha+\frac{1}{2}i_Yi_XH\\
        &-\langle D_\nabla X, \beta\rangle-\langle D_{\nabla^\dagger}\alpha-\frac{1}{2}i_XH, Y\rangle \\
         =&[X,Y]+\nabla_YX+\lie_X\beta-\langle D_\nabla X,\beta\rangle\\
         &-i_Yd\alpha+\nabla_Y^\dagger \alpha -\langle D_{\nabla^\dagger}\alpha, Y\rangle\\
         =&\nabla_X Y+\nabla^\dagger_X \beta.
\end{align*}
This calculation will be important in Section \ref{sec:characteristic}.
\end{example}

\begin{remark}\label{rmk:top}
$E$-connections can be extended to tensor powers using a derivation rule. In particular, the connection $\nabla^E$ extends to $\bigwedge^{\topp} E$ as follows:
\[ \nabla^E (e_1 \wedge \cdots \wedge e_m) = \sum_{i=1}^m e_1 \wedge \cdots \wedge \nabla^E_e e_i \wedge \cdots e_m.\]
It turns out that this extension of $\nabla^E$ coincides with the representation $\nabla^{\topp}$ in Example \ref{ex:wedgetop}. The reason is that, for $e, e_1, e_2 \in \Gamma(E)$,
\begin{align*}
\pair{\nabla^E_e e_1 - \cbrack{e}{e_1}}{e_2}  &= \pair{\nabla_{\rho(e_1)} e}{e_2} - \pair{\rho^*\pair{D_\nabla e}{e_1}}{e_2} \\
&= \pair{\nabla_{\rho(e_1)} e}{e_2} - \pair{\nabla_{\rho(e_2)} e}{e_1},
\end{align*}
which is $C^\infty(M)$-linear and skew-symmetric in $e_1, e_2$.  Therefore $\nabla_{\rho(e_1)}e - \rho^* \pair{D_\nabla e}{e_1}$ is a traceless endomorphism of $E$ which vanishes in the extension to $\bigwedge^{\topp}E$.
\end{remark}

\begin{remark}
The connections in Example \ref{repCou} are analogues of the \emph{Bott connections} for Lie algebroids \cite{fernandes:connections}. In further analogy with the case of Lie algebroids, one should expect the connections in Example \ref{repCou} to be part of an adjoint representation up to homotopy of $E$ on the $3$-term complex $T^*M \to E \to TM$. A full development of this idea would be outside the scope of this paper, but we will make some comments that may be helpful to the reader who wishes to pursue further.

Recall \cite{abad-crainic:algbd, m-gs:vbalgbds} that a representation up to homotopy of a Lie algebroid $A \to M$ on a graded vector bundle $\mathcal{B} \to M$ is a differential on the complex $\wedge \Gamma(A^*) \otimes \Gamma(\mathcal{B})$. Because this complex is bigraded, the differential splits into different components, and the equation $D^2 = 0$ splits into a series of equations relating the different components. The components can be interpreted as connections and Lie algebroid forms with values in homomorphism bundles. 

One could similarly define a representation up to homotopy of a Courant algebroid $E \to M$ as a differential on the complex $C^k(E) \otimes \Gamma(\mathcal{B})$. The rest of the analysis would be very similar to the case of Lie algebroids. In particular, \eqref{eqn:dcartan} could be used to express some of the equations in terms of brackets.
\end{remark}

The following proposition gives some useful properties of the Bott connections.
\begin{prop}\label{prop:selfadjoint}
Let $E \to M$ be a Courant algebroid, and let $\nabla$ be a linear connection on $E$. Then the $E$-connections $\nabla^E$, $\nabla^{TM}$, and $\nabla^{T^*M}$ satisfy the identities
\begin{align*}
	&\pair{\nabla^E_{e_1}e_2}{e_3}+\pair{e_2}{\nabla^E_{e_1}e_3}=\rho(e_1)\pair{e_2}{e_3},\\
	&\pair{\nabla^{TM}_{e}X}{\alpha}+\pair{X}{\nabla^{T^*M}_e\alpha}=\rho(e)\pair{X}{\alpha}.
\end{align*}
for all $e, e_i \in \Gamma(E)$, $X \in \vect(M)$, and $\alpha \in \Omega^1(M)$.
\end{prop}

\begin{proof}
We first note that $\pair{\rho^*\pair{D_\nabla e_1}{e_2}}{e_3} = \pair{\pair{D_\nabla e_1}{e_2}}{\rho(e_3)} = \pair{\nabla_{\rho(e_3)} e_1}{e_2}$. Using this and Axiom (C2) in Definition \ref{courant}, we have
\begin{align*}
		\pair{\nabla^E_{e_1}e_2}{e_3}+\pair{e_2}{\nabla^E_{e_1}e_3}=& \pair{\cbrack{e_1}{e_2}+\nabla_{\rho(e_2)}e_1}{e_3} -\pair{\nabla_{\rho(e_3)} e_1}{e_2}\\
		&+\pair{e_2}{\cbrack{e_1}{e_3}+\nabla_{\rho(e_3)}e_1}-\pair{\nabla_{\rho(e_2)} e_1}{e_3}\\
		=& \pair{\cbrack{e_1}{e_2}}{e_3} + \pair{e_2}{\cbrack{e_1}{e_3}}\\
		=& \rho(e_1)\pair{e_2}{e_3}.
\end{align*}
Additionally, since $\pair{X}{\lie_{\rho(e)}\alpha} = \rho(e)\pair{X}{\alpha} - \pair{[\rho(e),X]}{\alpha}$ and $\pair{X}{\pair{D_\nabla e}{\rho^*\alpha}} = \pair{\nabla_X \alpha}{\rho^*\alpha} = \pair{\rho(\nabla_X e)}{\alpha}$, we have
\begin{align*}
		\pair{\nabla^{TM}_e X}{\alpha}+\pair{X}{\nabla^{T^*M}_e\alpha} &= \pair{[\rho(e),X]+\rho(\nabla_X e)}{\alpha}+\pair{X}{\lie_{\rho(e)}\alpha-\pair{D_\nabla e}{\rho^*\alpha}}\\
		&=\rho(e)\pair{X}{\alpha}.\qedhere
\end{align*}
\end{proof}

\begin{remark}
The first identity in Proposition \ref{prop:selfadjoint} says that $\nabla^E$ is self-adjoint with respect to the bilinear form on $E$. The second identity can also be seen as a self-adjointness property, in the sense that the $E$-connection $\nabla^{TM} \oplus \nabla^{T^*M}$ on $TM \oplus T^*M$ is self-adjoint with respect to the usual symmetric bilinear form $\pair{(X,\alpha)}{(Y,\beta)} = \alpha(Y) + \beta(X)$. 
\end{remark}

\section{The modular class}\label{sec:modular}

The modular class of a Courant algebroid was introduced by Sti\'enon and Xu \cite{stienon-xu:modular} in the more general setting of Loday algebroids. Their construction produces a degree $1$ class in the \emph{na\"ive cohomology} of $E$. By working in the na\"ive complex, they had access to a Cartan formula, which allowed them to construct the class following essentially the same procedure as in \cite{elw}. In degree $1$, the na\"ive cohomology is isomorphic to the standard cohomology, so they indirectly obtained a class in $H^1(E)$. 

Here we review the construction. But now, as a result of Theorem \ref{thm:cartan}, we are able to make the additional observation that the construction directly produces a class in $H^1(E)$.

Let $E \to M$ be a Courant algebroid, and let $\nabla$ be a representation of $E$ on a trivializable real line bundle $L \to M$. Choose a nonvanishing section $\lambda\in\Gamma(L)$. Then, for each $e\in \Gamma(E)$, there exists a unique function $g_e \in C^\infty(M)$ such that $\nabla_e \lambda= g_e \lambda$. The fact that $\nabla_{fe}\lambda= f\nabla_e\lambda$ implies that $g_e=\pair{\xi_{\lambda}}{e}$ for some $\xi_\lambda\in\Gamma (E)$. 

The flatness of $\nabla$ implies that, for all $e_1, e_2 \in \Gamma(E)$,
\begin{align*}
		\pair{\xi_\lambda}{\cbrack{e_1}{e_2}}\lambda&=\nabla_{\cbrack{e_1}{e_2}}\lambda \\
		&=\nabla_{e_1}\nabla_{e_2}\lambda-\nabla_{e_2}\nabla_{e_1}\lambda \\
		&=\left(\rho(e_1)(\pair{\xi_\lambda}{e_2}) - \rho(e_2)(\pair{\xi_\lambda}{e_1})\right)\lambda.
\end{align*}
Therefore, if we view $\xi_\lambda$ as an element of $C^1(E)$, we have $d_E \xi_\lambda = 0$.

If $\lambda'$ is another nonvanishing section, then we can write $\lambda'=f\lambda$ for some nonvanishing $f \in C^\infty(M)$. Then we have
\[\pair{\xi_{\lambda'}}{e}f \lambda = \pair{\xi_{\lambda'}}{e}\lambda'=\nabla_e\lambda'=\nabla_e f\lambda=f\nabla_e \lambda+\rho(e)(f)\lambda,\]
and therefore 
\[\xi_{\lambda'}=\xi_\lambda+d_E \log|f|.\]
Thus the class $[\xi_\lambda] \in H^1(E)$ is well-defined and independent of $\lambda$ (so the subscript $\lambda$ can be omitted). The class $[\xi]$ is called the \emph{modular class} of the representation $(L,\nabla)$. 

In the case where $L$ is not trivializable, the modular class is defined as $\frac{1}{2}[\xi^{L \otimes L}]$, where $[\xi^{L \otimes L}]$ is the modular class of the induced representation on $L \otimes L$. This definition is justified by the following fact.
\begin{lemma}
Let $(L,\nabla^L)$ and $(L',\nabla^{L'})$ be trivializable line bundles with representations of $E$. Let $[\xi^L]$ and $[\xi^{L'}]$ denote the respective modular classes. Then the modular class of the induced representation on $L\otimes L'$ is $[\xi^L] + [\xi^{L'}]$.
\end{lemma}
\begin{proof}
Let $\lambda^L$ and $\lambda^{L'}$ be nonvanishing sections of $L$ and $L'$, respectively. Then
\begin{align*}
	\pair{\xi^{L\otimes L'}_{\lambda^L \otimes \lambda^{L'}}}{e}\lambda^L\otimes\lambda^{L'}&=\nabla^{L\otimes L'}_e(\lambda^L\otimes\lambda^{L'})\\
	&=\nabla^L_e\lambda^L\otimes\lambda^{L'}+\lambda^L\otimes\nabla_e^{L'}\lambda^{L'}\\
	&=\pair{\xi^L_{\lambda^L}+\xi^{L'}_{\lambda^{L'}}}{e}\lambda^L\otimes\lambda^{L'}. \qedhere
\end{align*}
\end{proof}

The (intrinsic) modular class of $E$ is defined in \cite{stienon-xu:modular} to be the modular class of the canonical representation $\nabla^{\topp}$ on $\bigwedge^{\topp} E$ (see Example \ref{ex:wedgetop}). They proved that, if $E = A \oplus A^*$, where $(A, A^*)$ is a Lie bialgebroid, then the modular class vanishes. It is also known (e.g.\ \cite{meinrenken:book}) that quadratic Lie algebras (i.e.\ Courant algebroids where $M$ is a point) are unimodular. Using a supergeometric argument, Grabowski \cite{grabowski:modular} observed that, in fact, all Courant algebroids have vanishing modular class. Here we give a new proof of this result.
\begin{prop}\label{prop:modtensor}
The modular class vanishes for every Courant algebroid.
\end{prop}
\begin{proof}
For any representation $\nabla$ of $E$ on a line bundle $L$, we can consider the dual representation $\nabla^*$ on $L^*$, given by
\[ \pair{\nabla^*_e \alpha}{\ell} = \rho(e)\pair{\alpha}{\ell} - \pair{\alpha}{\nabla_e \ell}\]
for $e \in \Gamma(E)$, $\ell \in \Gamma(L)$, and $\alpha \in \Gamma(L^*)$. Since $L \otimes L^*$ is canonically isomorphic to the trivial representation on $M \times \R$, Proposition \ref{prop:modtensor} implies that $[\xi^{L^*}] = -[\xi^L]$. 

Comparing with \eqref{eqn:adjoint}, we see that, if $L$ has a nondegenerate bilinear form, then $\nabla^\dagger$ is obtained by transferring $\nabla^*$ via the isomorphism $L \cong L^*$ associated to the bilinear form, so the modular classes associated to $\nabla$ and $\nabla^\dagger$ are minuses of each other. Thus, if the representation on $L$ is self-adjoint, then its modular class vanishes.

The result then follows from the fact that the canonical representation on $\bigwedge^{\topp} E$ is self-adjoint; this can be seen by a direct calculation, using Axiom (C2) in Definition \ref{courant}, or alternatively as a consequence of Remark \ref{rmk:top} and Proposition \ref{prop:selfadjoint}.
\end{proof}

\begin{remark}
Although Courant algebroids are unimodular, it should be emphasized that modular classes do not vanish in general. In fact, given any cocycle $\xi \in C^1(E) = \Gamma(E)$, one can use the equation $\nabla_e \lambda = \pair{\xi}{e}\lambda$ to define a representation on a trivial line bundle for which the modular class is $[\xi]$.
\end{remark}

\begin{remark}
The unimodularity of Courant algebroids has a simple explanation. From the definition, it is clear that the modular class is the obstruction to the existence of an invariant volume form. Because a Courant algebroid comes equipped with an invariant nondegenerate bilinear form, there is always an induced invariant volume form on $\bigwedge^{\topp} E$. In the next section, we construct higher characteristic classes which are obstructions to the existence of an invariant (positive definite) metric. It is known that there exist examples of Lie algebras that admit an invariant nondegenerate bilinear form but do not admit an invariant \emph{positive-definite} bilinear form. We expect that the higher characteristic classes can detect such phenomena.
\end{remark}

\section{Characteristic classes}\label{sec:characteristic}

In this section, we describe the construction of intrinsic characteristic classes associated to a Courant algebroid $E$. Surprisingly, the construction is even simpler than that of Lie algebroids (e.g.\ \cite{crainic-fernandes:secondary}) because it does not require a representation up to homotopy. Indeed, the results of previous sections allow for a construction that is similar to the classical theory, e.g.\ \cite{chern-simons}. We will sometimes omit proofs that closely resemble proofs in the classical theory.

\subsection{Chern forms and Chern-Simons forms}

Let $E \to M$ be a Courant algebroid, and let $\nabla$ be an $E$-connection on a vector bundle $B \to M$. We first recall a few facts about $\End(B)$-valued cochains.
\begin{itemize}
\item The elements of $C^k(E; \End(B))$ can be identified with degree $k$ operators on $C^\bullet(E;B)$ that are $C^\bullet(E)$-linear.
\item Composition of operators gives a product on $C^\bullet(E;\End(B))$.
\item The covariant derivative operator $D_{\widetilde{\nabla}}$ associated to the $E$-connection \eqref{eqn:connend} is given by the graded commutator of operators: $D_{\widetilde{\nabla}}\phi = [D_\nabla, \phi ]$ for $\phi \in C^\bullet(E;\End(B))$.
\item There is a natural trace map $\tr: C^\bullet(E; \End(B)) \to C^\bullet(E)$. Moreover, the identity
\begin{equation}\label{eqn:dtrace}
d_E \tr(\phi) = \tr(D_{\widetilde{\nabla}}\phi)
\end{equation}
 holds for all $\phi \in C^\bullet(E;\End(B))$.
\end{itemize}

For $k=1,2,\dots$, we define the Chern forms $\ch_k(\nabla) \in C^{2k}(E)$ by
\[ \ch_k(\nabla) = \tr(F_\nabla^k).\]
From \eqref{eqn:dtrace} and the Bianchi identity, we see that 
\begin{equation}\label{eqn:dch}
    d_E \ch_k(\nabla) = 0.
\end{equation}

Given a path $\nabla_t$ of $E$-connections on $B$, for $0 \leq t \leq 1$, we can define the Chern-Simons forms $\cs_k(\nabla_t) \in C^{2k-1}(E)$ by
\begin{equation} \label{eqn:cs} \cs_k(\nabla_t) = k \int_0^1 \tr\left(\frac{d}{dt}[\nabla_t] \cupprod F_{\nabla_t}^{k-1} \right) dt.
\end{equation}
The Chern-Simons forms are transgressions of the Chern forms, in the following sense.
\begin{prop}\label{prop:dcs}
Let $\nabla_t$ be a path of $E$-connections on $B$. Then
$d_E \cs_k(\nabla_t) = \ch_k(\nabla_1) - \ch_k(\nabla_0)$.
Therefore, the cohomology class $[\ch_k(\nabla)] \in H^{2k}(E)$ is independent of $\nabla$.
\end{prop}
\begin{proof}
By the Fundamental Theorem of Calculus, we have
\begin{align*}
    \ch_k(\nabla_1) - \ch_k(\nabla_0) &= \int_0^1 \frac{d}{dt}[\ch_k(\nabla_t)] dt \\
    &= \int_0^1 \frac{d}{dt}[\tr(F_{\nabla_t}^k)] dt \\
    &= k \int_0^1 \tr\left( \frac{d}{dt}[F_{\nabla_t}] \cupprod F_{\nabla_t}^{k-1}\right) dt.
\end{align*}
Using the fact that $\frac{d}{dt}[F_{\nabla_t}] = D_{\widetilde{\nabla}_t}\left(\frac{d}{dt}[\nabla_t]\right)$ and the Bianchi identity, we see that the above is
\begin{align*}
    k \int_0^1 \tr\left( D_{\widetilde{\nabla}_t}\left(\frac{d}{dt}[\nabla_t]\right) \cupprod F_{\nabla_t}^{k-1}\right) dt &= k \int_0^1 d_E \tr \left(\frac{d}{dt}[\nabla_t] \cupprod F_{\nabla_t}^{k-1}\right) dt \\
    &= d_E \cs_k(\nabla_t). \qedhere
\end{align*}
\end{proof}

\begin{prop}\label{prop:cshom}
Let $\nabla_t,\nabla'_t$ be two paths of $E$-connections on $B$ with $\nabla_0=\nabla'_0$ and $\nabla_1=\nabla'_1$. Then $cs_k(\nabla_t)-cs_k(\nabla'_t)$ is exact.
\end{prop}
\begin{proof}
Let $\bnabla = (1-s)\nabla_t + s\nabla'_t$. Then
\begin{align*}
    \cs_k(\nabla'_t) - \cs_k(\nabla_t) &= \int_0^1 \frac{d}{ds}\left[\cs_k(\bnabla)\right] ds \\
    &= \int_0^1 \frac{d}{ds} \left[ \int_0^1 k \tr \left(\frac{\partial}{\partial t}[\bnabla] \cupprod F_{\bnabla}^{k-1} \right) dt \right] ds \\
    &= k \int_0^1 \int_0^1 \tr \left( \frac{\partial^2}{\partial s \partial t}[\bnabla] \cupprod F_{\bnabla^{k-1}} + \frac{\partial}{\partial t}[\bnabla] \cupprod \frac{\partial}{\partial s} [F_{\bnabla}^{k-1}]\right)dt\,ds.
\end{align*}
Integrating by parts in the first term, we get
\begin{align*}
 & k \int_0^1 \int_0^1 \tr \left( \frac{\partial}{\partial t}[\bnabla] \cupprod \frac{\partial}{\partial s} [F_{\bnabla}^{k-1}] -  \frac{\partial}{\partial s}[\bnabla] \cupprod \frac{\partial}{\partial t} [F_{\bnabla}^{k-1}] \right)dt\,ds \\
 = & k \int_0^1 \int_0^1 \sum_{i=0}^{k-2} \tr \left( \frac{\partial}{\partial t}[\bnabla]  \cupprod F_{\bnabla}^i \cupprod D_{\widetilde{\bnabla}}\left(\frac{\partial}{\partial s}[\bnabla]\right) \cupprod F_{\bnabla}^{k-i-2} \right. \\
 & \left. - \frac{\partial}{\partial s}[\bnabla] \cupprod F_{\bnabla}^i \cupprod D_{\widetilde{\bnabla}}\left(\frac{\partial}{\partial t}[\bnabla]\right) \cupprod F_{\bnabla}^{k-i-2} \right) dt\,ds \\
 =& d_E k \int_0^1 \int_0^1 \sum_{i=0}^{k-2} \tr \left( \frac{\partial}{\partial s}[\bnabla] \cupprod F_{\bnabla}^i \cupprod \frac{\partial}{\partial t}[\bnabla] \cupprod F_{\bnabla}^{k-i-2}\right) dt\,ds,
\end{align*}
which is exact.
\end{proof}
It follows from Proposition \ref{prop:cshom} that, if $\nabla_0, \nabla_1$ are $E$-connections on $B$ such that $\ch_k(\nabla_0) = \ch_k(\nabla_1) = 0$ for some $k$, then we can unambiguously define $[\cs_k(\nabla_0,\nabla_1)] = [\cs_k(\nabla_t)] \in H^{2k-1}(E)$, where $\nabla_t$ is any path from $\nabla_0$ to $\nabla_1$. Furthermore, the triangle identity
\begin{equation}\label{eqn:triangle}
 [\cs_k(\nabla_0, \nabla_1)] + [\cs_k(\nabla_1,\nabla_2)] = [\cs_k(\nabla_0, \nabla_2)]
 \end{equation}
holds for all $\nabla_i$ such that $\ch_k(\nabla_i) = 0$.

\begin{remark}
Taking $\nabla_t$ to be the straight-line path between $\nabla_0$ and $\nabla_1$, we can explicitly evaluate the integral in \eqref{eqn:cs} to find formulas for $[\cs_k(\nabla_0,\nabla_1)]$. Specifically, if we write $\phi = \nabla_1 - \nabla_0$, then we can let $\nabla_t = \nabla_0 + t\phi$. Then the first two Chern-Simons forms are
\[ \cs_1(\nabla_t) = \int_0^1 \tr(\phi) dt = \tr(\phi)\]
and 
\begin{align*}
 \cs_2(\nabla_t) &= 2\int_0^1 \tr(\phi \cupprod F_{\nabla_t}) dt \\
 &= 2 \int_0^1 \tr\left( \phi \cupprod (F_{\nabla_0} + t D_{\tilde{\nabla}_0} \phi + t^2 \phi^2) \right) dt\\
 &= \tr\left( 2\phi \cupprod F_{\nabla_0} + \phi \cupprod D_{\tilde{\nabla}_0}\phi + \frac{2}{3} \phi^3\right).
 \end{align*}
 We note that, when $\nabla_0$ is flat, this formula for $\cs_2(\nabla_t)$ is (up to scalar) formally identical to the formula for the classical Chern-Simons $3$-form.
\end{remark}

\begin{remark}
The above construction fits into the general theory of transgressions as described in, e.g., \cite{crainic-fernandes:secondary, freed:cs2}. A rigorous treatment requires infinite-dimensional analysis, but the picture is quite nice, so we sketch it here.

We can view $\ch_k$ as being a $C^{2k}(E)$-valued function on the space of $B$-connections on $E$. From \eqref{eqn:dch}, we have $d_E \ch_k = 0$. The formula \eqref{eqn:cs} essentially defines a $C^{2k-1}(E)$-valued $1$-form $\alpha_k$ on the space of $B$-connections on $E$, such that 
\[\cs_k(\nabla_t) = \int_{\nabla_t} \alpha.\] 
Proposition \ref{prop:dcs} can then be interpreted as saying that $d_E \alpha_k = \delta \ch_k$, where $\delta$ is the de Rham operator. Thus, $\alpha_k$ is a transgression of $\ch_k$.

There is a $C^{2k-2}(E)$-valued $2$-form $\beta_k$ that is a transgression of $\alpha_k$. Its existence explains why Proposition \ref{prop:cshom} is true; indeed, a formula for $\beta_k$ can be extracted from the proof of Proposition \ref{prop:cshom}.

One could continue, constructing higher transgression forms such that $\ch_k + \alpha_k + \beta_k + \dots$ is closed in the total complex of $C(E)$-valued differential forms on the space of $B$-connections on $E$.
\end{remark}

\subsection{Gauge and metric compatibility}
Let $u$ be an automorphism of $B$ covering the identity map on $M$. Then $u$ acts on $E$-connections on $B$ by gauge transformations, as follows:
\[ u(\nabla)_e b = (u \circ \nabla_e \circ u^{-1}) b.\]
The Chern-Simons forms satisfy the following ``small gauge-invariance'' property.
\begin{prop}\label{prop:gauge}
Let $\nabla$ be an $E$-connection on $B$, and let $u_t$ be a path of automorphisms of $B$ covering the identity map on $M$, with $u_0 = \id$. Then $\cs_k(u_t(\nabla))$ is exact.
\end{prop}
\begin{proof}
Clearly, $F_{u_t(\nabla)} = u_t F_\nabla u_t^{-1}$. Additionally, writing $\frac{d}{dt}[u_t] = \dot{u}_t$, we have
\begin{align*}
\frac{d}{dt}\left[u_t(\nabla)\right] &= \frac{d}{dt} [ u_t D_\nabla u_t^{-1} ] \\
&= \dot{u}_t D_\nabla u_t^{-1} - u_t D_\nabla u_t^{-1} \dot{u}_t u_t^{-1} \\
&= [\dot{u}_t u_t^{-1}, u_t D_\nabla u_t^{-1}].
\end{align*}
Using these calculations and the Bianchi identity, we see that
\begin{align*}
\cs_k(u_t(\nabla)) &= k \int_0^1 \tr\left( [\dot{u}_t u_t^{-1}, u_t D_\nabla u_t^{-1}] \cupprod u_t F_\nabla^{k-1} u_t^{-1} \right) dt \\
&= k \int_0^1 \tr \left( u_t [u_t^{-1} \dot{u}_t, D_\nabla] \cupprod F_\nabla^{k-1} u_t^{-1} \right) dt\\
&= -k \int_0^1 \tr \left( [D_\nabla, u_t^{-1} \dot{u}_t F_\nabla^{k-1}] \right) dt\\
&= -k d_E \int_0^1 \tr \left( u_t^{-1} \dot{u}_t F_\nabla^{k-1} \right) dt,
\end{align*}
which is exact.
\end{proof}

Suppose that $B$ is equipped with a nondegenerate bilinear form $\pair{\cdot}{\cdot}_B$. Then the Chern-Simons forms satisfy the following compatibility condition with respect to the adjoint operation (see \eqref{eqn:adjoint}).
\begin{prop}\label{prop:csadjoint}
For any path $\nabla_t$ of $E$-connections on $B$, 
\[ \cs_k(\nabla_t^\dagger) = (-1)^k \cs_k(\nabla_t).\]
\end{prop}
\begin{proof}
From Proposition \ref{prop:adjointcurve}, we have $F_{\nabla_t^\dagger} = - (F_{\nabla_t})^\dagger$. Also, from \eqref{eqn:adjoint} we see that $\frac{d}{dt}[\nabla_t^\dagger] = -\frac{d}{dt}[\nabla_t]^\dagger$. Putting these into the definition of the Chern-Simons forms, we immediately obtain the result.
\end{proof}

\subsection{Secondary characteristic classes}

Recall from Example \ref{repCou} that, given a choice of linear connection $\nabla: \vect(M) \times \Gamma(E) \to \Gamma(E)$, there is an induced $E$-connection $\nabla^E$ on $E$.

\begin{prop}\label{prop:noch}
If $k$ is odd, then $\ch_k(\nabla^E) = 0$.
\end{prop}
\begin{proof}
A straightforward calculation using Proposition \ref{prop:selfadjoint} shows that 
\begin{equation}\label{eqn:skewf}
\pair{F_{\nabla^E}(e_1,e_2) e_3}{e_4} = \pair{e_3}{-F_{\nabla^E}(e_1,e_2)e_4}
\end{equation}
for all $e_i \in \Gamma(E)$, so $F_{\nabla^E}$ is skew-symmetric with respect to $\pair{\cdot}{\cdot}$. Therefore, $\tr(F_{\nabla^E}^k) = (-1)^k \tr(F_{\nabla^E}^k)$, and the result follows.
\end{proof}

To obtain Chern-Simons forms, we need another $E$-connection. To do this, we choose a (positive definite) metric $g(\cdot,\cdot)$ on $E$ and let $\nabla^{E,g}$ be the adjoint with respect to $g$, given by
\[ g(\nabla^E_{e_1} e_2, e_3) + g(e_2, \nabla^{E,g}_{e_1} e_3) = \rho(e_1)g(e_2,e_3) \]
for $e_i \in \Gamma(E)$. Note that we do not assume any compatibility condition between $g$ and the Courant bracket or anchor.

From Proposition \ref{prop:adjointcurve}, we have $F_{\nabla^{E,g}} = -(F_\nabla)^g$, where the $g$ superscript denotes the adjoint with respect to $g$. It then follows from Proposition \ref{prop:noch} that $\ch_k(\nabla^{E,g}) = 0$ for odd $k$. Proposition \ref{prop:cshom} then implies that $[\cs_k(\nabla^E, \nabla^{E,g})] \in H^{2k-1}(E)$ is a well-defined cohomology class for odd $k$. 

\begin{definition}\label{dfn:intrinsic}
The \emph{secondary characteristic classes} of $E$ are the classes 
\[[\cs_{2k-1}(\nabla^E, \nabla^{E,g})] \in H^{4k-3}(E).\]
\end{definition}
A priori, the secondary characteristic classes depend on the choices of a linear connection $\nabla$ and a metric $g$ on $E$. The following theorem shows that the classes don't depend on these choices, so they are intrinsically defined.

\begin{thm}
The cohomology class $[\cs_k(\nabla^E, \nabla^{E,g})]$ is independent of the choice of linear connection $\nabla$ and metric $g$.
\end{thm}
\begin{proof}
Let $\nabla^{E'}$ be the $E$-connection arising from a different choice of linear connection, let $\phi = \nabla^{E'} - \nabla^{E}$, and let $\nabla^E_t = \nabla^E + t\phi$. From Proposition \ref{prop:selfadjoint}, we see that $\pair{\phi_{e_1} e_2}{e_3} = - \pair{e_2}{\phi_{e_1} e_3}$,
so $\phi$ is skew-symmetric with respect to $\pair{\cdot}{\cdot}$. 

The dependence of $\nabla^E$ on $\nabla$ is affine, so $\nabla^E_t$ is the $E$-connection associated to some choice of linear connection for all $t$. Therefore, from \eqref{eqn:skewf}, we see that, when $k$ is odd, $F^{k-1}_{\nabla^E_t}$ is symmetric with respect to $\pair{\cdot}{\cdot}$ for all $t$. 

A basic fact from linear algebra is that the trace of the product of a skew-symmetric matrix with a symmetric matrix vanishes. Thus, $\tr(\phi \cupprod F^{k-1}_{\nabla^E_t}) = 0$, so $\cs_k(\nabla^E_t)$ vanishes. By \eqref{eqn:triangle} and Proposition \ref{prop:csadjoint}, we deduce that $[\cs_k(\nabla^E, \nabla^{E,g})]$ is independent of the choice of linear connection. 

Now let $g_t$ be a path of metrics with $g_0 = g$. Then we can obtain a path $u_t$ of automorphisms of $E$, given by $g_t(e_1,e_2) = g(u_t(e_1),e_2)$ for all $e_1,e_2 \in \Gamma(E)$. 

A straightforward calculation shows that $\nabla^{E,g_t} = u_t (\nabla^{E,g})$. Since $\cs_k(\nabla^{E,g_t}) = \cs_k(u_t(\nabla^{E,g}))$ is exact by Proposition \ref{prop:csadjoint}, we conclude by \eqref{eqn:triangle} that $[\cs_k(\nabla^E, \nabla^{E,g})]$ is independent of the choice of metric.
\end{proof}

\begin{example}
In the case where $M$ is a point, so that $E = \mathfrak{g}$ is a quadratic Lie algebra, there is a unique linear connection, and the associated $E$-connection $\nabla^E = \ad$ is just the adjoint representation. In this case it can be seen that the secondary characteristic classes have natural representatives that are independent of the choice of metric $g$ and given (up to a constant factor) by
\[ \cs_k(\ad,\ad^g) = \tr(\ad^{2k-1}).\]
Thus, as a cocycle,
\[ \cs_k(\ad,\ad^g)(e_1,\dots,e_{2k-1}) = \sum_{\pi \in S_{2k-1}}\sgn(\pi)\tr\left(\ad_{e_{\pi(1)}} \cdots \ad_{e_{\pi(2k-1)}}\right)\]
for $e_i \in \mathfrak{g}$. In the context of Poisson geometry, a similar formula appeared (with signs missing) in \cite{fernandes:connectionspoisson}. Note that, in this case, the bilinear form does not play a role in defining the secondary characteristic classes.
\end{example}

\begin{example}
Consider the case where $E=TM\oplus T^*M$, with an $H$-twisted Courant bracket; see Example \ref{standard-CA}.

Choose a Riemannian metric $g$ on $M$ and denote by $\nabla$ its Levi-Civita connection. Then $g+g^\dagger$ defines a metric on $E$. Since $\nabla$ is torsion-free, we can apply Example \ref{H-Econ} to obtain a linear connection on $E$ such that the associated $E$-connection is given by
\begin{equation*}\nabla^E_{X+\alpha}(Y+\beta)=\nabla_X Y+\nabla^\dagger_X\beta\end{equation*}
for  $X,Y\in\mathfrak{X}(M)$ and $\alpha,\beta\in\Omega^1(M).$ From this expression and the fact that $\nabla$ preserves the metric $g$ we obtain that
$\nabla^{E,g}=\nabla^E$. It immediately follows that all the secondary characteristic classes of $E$ vanish.
\end{example}

\section{Dirac structures}\label{sec:dirac}

We conclude the paper with some remarks about Dirac structures.

Let $E \to M$ be a Courant algebroid. A \emph{Dirac structure} in $E$ is a subbundle $L \subseteq E$ such that $L^\perp = L$ and $\cbrack{\Gamma(L)}{\Gamma(L)} \subseteq \Gamma(L)$. If $L \subseteq E$ is a Dirac structure, then the restriction of the Courant bracket to $L$ is a Lie bracket, giving $L \to M$ the structure of a Lie algebroid. Thus there is an associated cochain complex $\left(\bigwedge\Gamma(L^*), d_L\right)$, where $d_L$ is given by a Cartan formula. The cohomology of this complex is, by definition, the Lie algebroid cohomology $H^\bullet(L)$. There are also secondary characteristic classes (e.g.\ \cite{crainic-fernandes:secondary}) associated to $L$. These classes are elements of $H^{4k-3}(L)$ for $k \geq 1$. The first of these classes is the modular class \cite{elw} of $L$.

It should be emphasized that $H^\bullet(L)$ and the characteristic classes therein are defined intrinsically with respect to the Lie algebroid structure of $L$, so they do not contain any information about how $L$ sits inside of $E$.

Consider a $k$-cochain $\omega \in C^k(E)$. From Definition \ref{cochain} we see that, if $L \subseteq E$ is a Dirac structure, then the restriction of $\omega$ to $\Gamma(L)$ is skew-symmetric. Thus there is a natural map $\pi: C^\bullet(E) \to \bigwedge^\bullet\Gamma(L^*)$. As an immediate consequence of Theorem \ref{thm:cartan}, we see that $\pi$ is compatible with the differentials and therefore induces a map $\pi^*: H^\bullet(E) \to H^\bullet(L)$.

The map $\pi$ can provide information about the relationship between $L$ and $E$. For example, for each $k$, we can define \emph{relative characteristic classes} of $L$ as the difference between the characteristic class in $H^{4k-3}(L)$ and the image under $\pi^*$ of the characteristic class in $H^{4k-3}(E)$. These classes do not automatically vanish; in particular, when $k=1$, it follows from Proposition \ref{prop:modtensor} that the relative characteristic class is just the modular class of $L$.

We note that, in the special case of a projectible Courant algebroid, a relative modular class of a Dirac structure was defined by Grabowski \cite{grabowski:modular}. His definition uses projectibility in a nontrivial way and is not simply a special case of the relative characteristic classes defined here.

\bibliography{courantbib}
\bibliographystyle{amsplain}
\end{document}